\documentclass[submission,copyright,creativecommons]{eptcs}
\providecommand{\event}{GANDALF 2016} 
\usepackage{breakurl}             

\usepackage{amsmath}
\usepackage{amssymb}
\usepackage{amsthm}
\usepackage{tikz,wrapfig}
\usepackage{amscd}

\newtheorem{definition}{Definition}
\newtheorem{proposition}[definition]{Proposition}
\newtheorem{lemma}[definition]{Lemma}
\newtheorem{theorem}[definition]{Theorem}
\newtheorem{corollary}[definition]{Corollary}

\newtheorem{example}[definition]{Example}
\newtheorem{observation}[definition]{Observation}

\newcommand{\name}[1]{\textsc{#1}}
\newcommand{\hide}[1]{}

\title{A Semi-Potential for Finite and Infinite Sequential Games\\(Extended Abstract)}

\author{St\'{e}phane Le Roux
\institute{D\'epartement d'informatique\\ Universit\'e libre de
Bruxelles, Belgique}
\email{Stephane.Le.Roux@ulb.ac.be}
\and
Arno Pauly
\institute{D\'epartement d'informatique\\ Universit\'e libre de
Bruxelles, Belgique}
\email{\quad Arno.Pauly@cl.cam.ac.uk}
}

\def\titlerunning{A Semi-Potential for Sequential Games}
\def\authorrunning{S. Le Roux \& A. Pauly}

\begin{document}

 \maketitle

\begin{abstract}
We consider a dynamical approach to sequential games. By restricting the convertibility relation over strategy profiles, we obtain a semi-potential (in the sense of Kukushkin), and we show that in finite games the corresponding restriction of better-response dynamics will converge to a Nash equilibrium in quadratic time. Convergence happens on a per-player basis, and even in the presence of players with cyclic preferences, the players with acyclic preferences will stabilize. Thus, we obtain a candidate notion for rationality in the presence of irrational agents. Moreover, the restriction of convertibility can be justified by a conservative updating of beliefs about the other players strategies.

For infinite sequential games we can retain convergence to a Nash equilibrium (in some sense), if the preferences are given by continuous payoff functions; or obtain a transfinite convergence if the outcome sets of the game are $\Delta^0_2$-sets.
\end{abstract}

\section{Introduction}
Nash equilibria are the terminal strategy profiles of the better (or best) response dynamics. In general, though, these dynamics do not terminate (and do not even converge). A particular exception is found in the potential games \cite{shapley}: A potential is an acyclic joint order extension of the individual players improvement relations. In a potential game, better-response dynamics thus will always improve the potential, and hence terminates (at a Nash equilibrium) if the game is finite.

The notion of semi-potential was introduced by \name{Kukushkin} (\cite{kukushkin2}, also \cite{kukushkin3}) in order to salvage some of the nice properties of potential games for a larger class of games. Here, the players freedom to change their strategies is restricted -- however, only in such a way that if they can change the current outcome to a particular one, they can also do so in a way that is consistent with the restriction. In a \emph{generic} normal form game this is equivalent to a potential, as there different strategies will induce different outcomes. Nevertheless, several classes of non-generic games have no potential but have a semi-potential: \cite[Theorem 3]{kukushkin2} proved that it is the case for sequential games.

We study the restriction of the convertibility relation (we call it \emph{lazy convertibility}) as well as the resulting better-response dynamics (\emph{lazy improvement}) in some more detail. We give two alternative proofs of the termination at a Nash equilibrium in finite games, one of which yields a tight quadratic bound on the number of steps required. Moreover, our two proofs of termination works on a per player basis: Thus, any player with acyclic preferences will converge, even in the presence of players with cyclic preferences. Then we explore several possible extensions to infinite sequential games, and one to directed acyclic graphs. A very specific infinite setting was explored in \cite{boros}, and lazy improvement in infinite sequential games with continuous payoff functions was investigated by the authors in \cite{paulyleroux2}.

Some relevant properties of lazy improvement are:
\begin{itemize}
\item The dynamics are uncoupled: Each player bases her decisions only on her own preference and the strategies of the other players, but does not need to know the other players' preferences.
\item The dynamics are history-independent: Unlike e.g.~fictitious play or typical regret-minimization approaches (e.g.~\cite{hart}), the next step in the dynamics depends only on the current strategies of the players. In particular, players do not need additional memory for \emph{learning}.
\item We consider pure strategies, not stochastic ones. Thus, our approach has a very different flavour from the usual evolutionary game theory one (e.g.~\cite{cressman,cressman2,ziboxu}).
\item No restrictions akin to \emph{generic payoffs} are required, we merely need acyclic preferences to guarantee termination at a Nash equilibrium in finite games (and anyway this requirement cannot be avoided for existence of Nash equilibrium \cite{SLR-PhD08,SLR09}).
\item In a finite game, the dynamics stabilizes at a Nash equilibrium after a quadratic number of steps.
\item The stabilization result for the rational players, \textit{i.e.} with acyclic preferences, remains unaffected, if unpredictable players, \textit{i.e.} with cyclic preferences, are added.
\item Under some conditions, even in an infinite sequential game we can ensure stabilization at a Nash equilibrium after a transfinite number of steps.
\end{itemize}

The rest of the paper is organized as follows: Section \ref{sec:definition} introduces the core concept of \emph{lazy improvement}. Section \ref{sec:finite} proves that in a finite game, lazy improvement terminates at a Nash equilibrium. Section~\ref{sect:finite-2nd-proof} gives an alternative proof also showing that termination occurs after a quadratic number of improvement steps. Section~\ref{sec:beliefs} gives a basic epistemic justification for lazy convertibility. In Section \ref{sec:infinite} we discuss extensions to infinite games. Finally, Section \ref{sec:counter} provides a number of (counter)examples showing that, to some extent, our definitions have to be the way they are. An extended preprint is available as \cite{leroux2}.

\section{Background and Notation}
We start by introducing notation for both games in normal form and sequential games.

\begin{definition}[Games in normal form]\label{defn:gnf}
A game in normal form is a tuple $\langle A,(S_a)_{a\in A},O,v,(\prec_a)_{a\in A}\rangle$ satisfying the following:
\begin{itemize}
\item $A$ is a non-empty set (of players, or agents),
\item $\prod_{a\in A}S_a$ is a non-empty Cartesian product (whose elements are the \emph{strategy profiles} and where $S_a$ represents the strategies available to player $a$),
\item $O$ is a non-empty set (of possible outcomes),
\item $v:\prod_{a\in A} S_a\to O$ (the outcome function that values the strategy profiles),
\item Each $\prec_a$ is a binary relation over $O$ (modelling the preference of player $a$).
\end{itemize}
\end{definition}

\begin{definition}[Nash equilibrium]\label{defn:ne}
Let $\langle A,(S_a)_{a\in A} ,O,v,(\prec_a)_{a\in A}\rangle$ be a game in normal form. A strategy profile (profile for short) $s$ in $S:=\prod_{a\in A} S_a$ is a Nash equilibrium if it makes every player $a$ stable, \textit{i.e.} $v(s)\not\prec_a v(s')$ for all $s'\in S$ that differ from $s$ at most at the $a$-component.
\[NE(s)\quad:=\quad\forall a\in A,\forall s'\in S,\quad\neg(v(s)\prec_a v(s')\,\wedge\,\forall b\in A-\{a\},\,s_b= s'_b)\]
\end{definition}

Implicit in the concept of Nash equilibrium is the notion of \emph{convertibility}: An agent can convert one strategy profile to another, if they differ only in her actions. As lazy improvement will be introduced in Section~\ref{sec:definition} by restricting the convertibility relation, we provide a formal definition:

\begin{definition}[Convertibility, induced preference over profiles, and improvement]\label{defn:asyn-improv}\hfill
\begin{itemize}
\item Let $\langle A,(S_a)_{a\in A},O,v,(\prec_a)_{a\in A}\rangle$ be a game in normal form. For $s, s' \in \prod_{a\in A}S_a$, let $s\stackrel{c}{\twoheadrightarrow}_as'$ denote the ability of player $a$ to convert $s$ to $s'$ by changing her own strategy, formally $s\stackrel{c}{\twoheadrightarrow}_as':=\forall b\in A-\{a\},\,s_b=s'_b$.

\item Given a game $\langle A,(S_a)_{a\in A},O,v,(\prec_a)_{a\in A}\rangle$, let $s\prec_a s'$ denote $v(s)\prec_a v(s')$. So in this article $\prec_a$ may also refer to the induced preference over the profiles.

\item Let $\twoheadrightarrow_a\,:=\,\prec_a\cap\stackrel{c}{\twoheadrightarrow}_a$ be the individual improvement relations of the players and let $\twoheadrightarrow\,:=\,\cup_{a\in A}\twoheadrightarrow_a$ be the collective\footnote{The word \emph{collective} should not be misread to indicate any form of coordination between the players.} improvement relation.
\end{itemize}
\end{definition}

Observation~\ref{obs:ne-sink} below is a direct consequence of Definitions~\ref{defn:ne} and \ref{defn:asyn-improv}.

\begin{observation}\label{obs:ne-sink}
The Nash equilibria of a game are exactly the sinks, \textit{i.e.}, the terminal profiles of the collective improvement $\twoheadrightarrow$.
\end{observation}

A (generalized) \emph{potential} is an acyclic relation containing $\twoheadrightarrow$. Clearly a game has a potential iff  $\twoheadrightarrow$  itself is acyclic. If $\prod_{a\in A}S_a$ is finite, this is equivalent to the better-response dynamics' always converging. A less restrictive notion is a semi-potential (introduced in \cite{kukushkin2}). A semi-potential is an acyclic relation $\hookrightarrow$ contained in $\twoheadrightarrow$, such that whenever $s \twoheadrightarrow s'$ then there is some $s''$ with $s \hookrightarrow s''$ and $v(s) = v(s'')$. In words, if a strategy profile can be reached by an improvement step, then there is an equivalent strategy profile (w.r.t~outcome) reachable via a step in the semi-potential. It follows that the sinks of a semi-potential are exactly the sinks of the collective improvement. Thus, in a finite setting, the existence of a semi-potential in particular implies the existence of sinks, i.e.~Nash equilibria.

Our setting will be sequential games, rather than games in normal form. The idea here is that the players collectively choose a path through a tree, with each player deciding the direction at the vertices that she is controlling. The preferences refer only to the path created, choices off the chosen path are irrelevant. Thus, the evaluation map $v$ is highly non-injective\footnote{}, which in turn gives room for the notion of a semi-potential to be interesting. Formally, we define sequential games as follows:

\begin{definition}[Sequential games]
A sequential game is a tuple $(A, T, O, d, v, (\prec_a)_{a \in A})$ where
\begin{itemize}
\item $A$ is the non-empty set of players,
\item $T$ is a rooted tree (finite or infinite),
\item $O$ is the non-empty set of outcomes,
\item $d$ associates a player with each vertex in the tree,
\item $v$ associates an outcome with each maximal path from the root through the tree,
\item and for each player $a \in A$, $\prec_a$ is a relation on $O$ (the preference relation of $a$).
\end{itemize}
\end{definition}

The corresponding game in normal form is obtained as follows: Let a strategy of player $a$ associate an outgoing edge with each vertex controlled by $a$. If a strategy per player is given, the collective choices identify some maximal path $p$ through the tree, called the \emph{induced play}. Applying $v$ to that path yields the outcome of the game; i.e.~the valuation of the game in normal form is the composition of the map that identifies the induced play and the valuation of the sequential game.

In our concrete examples, the outcomes will be tuples of natural numbers, and the $n$-th player will prefer a tuple $(x_1,\ldots,x_{|A|})$ to $(y_1,\ldots,y_{|A|})$ iff $x_n > y_n$.

\section{Defining \emph{lazy improvement}}
\label{sec:definition}
The idea underlying lazy improvement is that we do not let a player change their irrelevant choices, i.e.~those choices not along the play induced after the improvement. Equivalently, we require a player to change as few choices as possible when changing the induced play.


\begin{definition}[Lazy convertibility and improvement]\label{defn:lazy-conv}\hfill
\begin{itemize}
\item For two strategy profiles $s$, $s'$ in a sequential game let $s\stackrel{c}{\rightharpoonup}_a s'$ (read: \emph{$a$ can lazily convert $s$ into $s'$}), if for any vertex $t \in T$, if $s(t) \neq s'(t)$, then $d(t) = a$ and $t$ lies along the play induced by $s'$.
\item Let $\rightharpoonup_a\,:=\,\prec_a\cap\stackrel{c}{\rightharpoonup}_a$ be the lazy improvement of player $a$ and let $\rightharpoonup\,:=\,\cup_{a\in A}\rightharpoonup_a$ be the (collective) lazy improvement.
\end{itemize}
\end{definition}

Let us exemplify the notion of lazy convertibilty, which has nothing to do with the preferences or the outcomes: player $a$ can lazily convert the leftmost strategy profile below into each of the profiles below, but not into any other profile. Player $a$ is written bold face at nodes where changes occur, and double lines represent strategy choices.

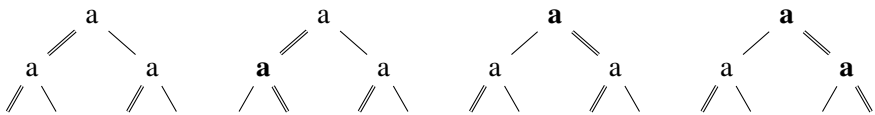
\begin{tabular}{cccc}
\begin{tikzpicture}[level distance=7mm]
\node{a}[sibling distance=16mm]
	child{node{a}[sibling distance=8mm] edge from parent[double]
		child{node{}edge from parent[double]}
		child{node{}}
	}
	child{node{a}[sibling distance=8mm]
			child{node{}edge from parent[double]}
			child{node{}}
	};
\end{tikzpicture}
&
\begin{tikzpicture}[level distance=7mm]
\node{a}[sibling distance=16mm]
	child{node{\bf{a}}[sibling distance=8mm] edge from parent[double]
		child{node{}}
		child{node{}edge from parent[double]}
	}
	child{node{a}[sibling distance=8mm]
			child{node{}edge from parent[double]}
			child{node{}}
	};
\end{tikzpicture}
&
\begin{tikzpicture}[level distance=7mm]
\node{\bf{a}}[sibling distance=16mm]
	child{node{a}[sibling distance=8mm]
		child{node{}edge from parent[double]}
		child{node{}}
	}
	child{node{a}[sibling distance=8mm]edge from parent[double]
			child{node{}edge from parent[double]}
			child{node{}}
	};
\end{tikzpicture}
&
\begin{tikzpicture}[level distance=7mm]
\node{\bf{a}}[sibling distance=16mm]
	child{node{a}[sibling distance=8mm]
		child{node{}edge from parent[double]}
		child{node{}}
	}
	child{node{\bf{a}}[sibling distance=8mm]edge from parent[double]
			child{node{}}
			child{node{}edge from parent[double]}
	};
\end{tikzpicture}
\end{tabular}

\noindent Contrary to the convertibility relations $\stackrel{c}{\twoheadrightarrow}_a$ which are equivalence relations, the lazy convertibility relations $\stackrel{c}{\rightharpoonup}_a$ are certainly reflexive but in general neither symmetric nor transitive. For instance, player $a$ cannot lazily convert the rightmost profile above back into the leftmost one. In the additional example below, player $a$ can convert the leftmost profile to the middle profile but not to the rightmost profile.

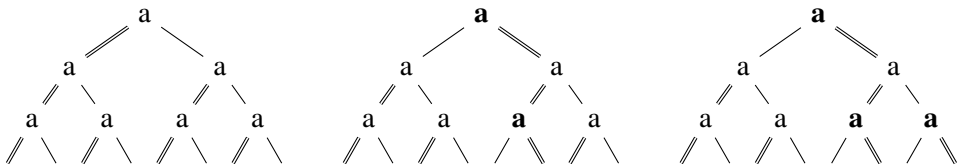
\begin{tabular}{ccc}
\begin{tikzpicture}[level distance=7mm]
\node{a}[sibling distance=20mm]
	child{node{a}[sibling distance=10mm]edge from parent[double]
		child{node{a}[sibling distance=8mm]edge from parent[double]
			child{node{}edge from parent[double]}
			child{node{}}
		}
		child{node{a}[sibling distance=8mm]
			child{node{}edge from parent[double]}
			child{node{}}
		}
	}
	child{node{a}[sibling distance=10mm]
		child{node{a}[sibling distance=8mm]edge from parent[double]
			child{node{}edge from parent[double]}
			child{node{}}
		}
		child{node{a}[sibling distance=8mm]
			child{node{}edge from parent[double]}
			child{node{}}
		}
	};
\end{tikzpicture}
&
\begin{tikzpicture}[level distance=7mm]
\node{\bf{a}}[sibling distance=20mm]
	child{node{a}[sibling distance=10mm]
		child{node{a}[sibling distance=8mm]edge from parent[double]
			child{node{}edge from parent[double]}
			child{node{}}
		}
		child{node{a}[sibling distance=8mm]
			child{node{}edge from parent[double]}
			child{node{}}
		}
	}
	child{node{a}[sibling distance=10mm] edge from parent[double]
		child{node{\bf{a}}[sibling distance=8mm]edge from parent[double]
			child{node{}}
			child{node{}edge from parent[double]}
		}
		child{node{a}[sibling distance=8mm]
			child{node{}edge from parent[double]}
			child{node{}}
		}
	};
\end{tikzpicture}
&
\begin{tikzpicture}[level distance=7mm]
\node{\bf{a}}[sibling distance=20mm]
	child{node{a}[sibling distance=10mm]
		child{node{a}[sibling distance=8mm]edge from parent[double]
			child{node{}edge from parent[double]}
			child{node{}}
		}
		child{node{a}[sibling distance=8mm]
			child{node{}edge from parent[double]}
			child{node{}}
		}
	}
	child{node{a}[sibling distance=10mm] edge from parent[double]
		child{node{\bf{a}}[sibling distance=8mm]edge from parent[double]
			child{node{}}
			child{node{}edge from parent[double]}
		}
		child{node{\bf{a}}[sibling distance=8mm]
			child{node{}}
			child{node{}edge from parent[double]}
		}
	};
\end{tikzpicture}
\end{tabular}

The lazy convertibility enjoys a useful property nonetheless, that the usual convertibility does not: if a player avoids a play during a sequence of lazy conversion, only the very same player is later able to make the last step to induce the same play again, possibly induced by a different profile. This phenomenon is more formally stated by Lemma~\ref{lem:avoid-play}.

\begin{lemma}\label{lem:avoid-play}
If $s\stackrel{c}{\rightharpoonup}_as_0\stackrel{c}{\rightharpoonup}\dots\stackrel{c}{\rightharpoonup}s_n\stackrel{c}{\rightharpoonup}_bs'$ where $s$ and $s'$ induce the same play, and if this play is different from the plays that are induced by the $s_i$, then $a=b$.
\end{lemma}

\begin{proof}
Let us prove the claim by induction on the underlying game. Since the play induced by $s_0$ is different from the play induced by $s$, these profiles are not just leaves, but proper trees instead. During the assumed $\stackrel{c}{\rightharpoonup}$ reduction of $s$, its subprofile that is chosen by the root owner in $s$ undergoes a $\stackrel{c}{\rightharpoonup}$ reduction too, say $t\stackrel{c}{\rightharpoonup}_at_0\stackrel{c}{\rightharpoonup}\dots\stackrel{c}{\rightharpoonup}t_n\stackrel{c}{\rightharpoonup}_bt'$, where $t$ and $t'$ induce the same play (and the root owner chooses $t'$ in $s'$). If all these subprofiles are equal, player $a$ must be the root owner (of $s$), since $s$ and $s_0$ induce different plays by assumption, and $b$ is also the root owner since $s_n$ and $s'$ induce different plays, so $a=b$. Now let $t_j$ be the first subprofile different from $t$, so $t_j$ induces a play different from $t$ and $t'$. For  all $k$ such that $j \leq k < n$, if $t_k$ and $t'$ induce different plays but $t_{k+1}$ and $t'$ induce the same play, then $s_{k+1}$ and $s'$ induce the same play by definition of $\stackrel{c}{\rightharpoonup}$, contradiction with the assumptions of the lemma, so all $t_j,\dots t_n$ induce plays different from that of $t'$. If $t_1\neq t$, then $a=b$ by the induction hypothesis, else $a$ must be the root owner and does not choose $t_1$ in $s_1$. The first time that $a$ chooses some $t_i$ again must be in $s_j$: indeed if it were before, $s$ and $s_i$ would induce the same play, and if it were after, $a$ could not change $t_{j-1}$ into $t_j$. Therefore $t_{j-1}\stackrel{c}{\rightharpoonup}_at_j\stackrel{c}{\rightharpoonup}\dots\stackrel{c}{\rightharpoonup}t_n\stackrel{c}{\rightharpoonup}_bt'$ and $a=b$ by the induction hypothesis.
\end{proof}

Despite the restrictive property from Lemma~\ref{lem:avoid-play}, the lazy convertibility is as effective as the usual convertibility, in the same sense as used in the definition of a semi-potential. (Thus, it will only remain to prove that lazy improvement is acyclic in order to establish lazy improvement as a semi-potential).

\begin{observation}
\label{obs:outcomes}
If $s \twoheadrightarrow s'$, then there is some strategy profile $s''$ such that $s \rightharpoonup s''$ and $v(s') = v(s'')$.
\begin{proof}
By definition, lazy convertibility does not restrict the choice of the new induced play, merely the ability to alter the strategy off the new induced play.
\end{proof}
\end{observation}

\begin{corollary}\label{corr:nash-lazy-term}
The Nash equilibria of a game are exactly the terminal profiles of the lazy improvement $\rightharpoonup$.
\end{corollary}

\section{Termination in finite games, first proof}
\label{sec:finite}

\begin{theorem}\label{thm:lazy-term}
Consider a sequential game played on a finite tree, and some sequence $(s_n)_{n \in \mathbb{N}}$ such that $s_n \rightharpoonup s_{n+1}$ for all $n \in \mathbb{N}$. Assume that for a player $a$ there are infinitely many $n$ with $s_n \rightharpoonup_a s_{n+1}$. Then $a$ has a cyclic preference.
\end{theorem}

\begin{proof}
Towards a contradiction let us assume that $a$'s preference is acyclic. Among the profiles $s$ such that $s = s_n \rightharpoonup_a s_{n+1}$ for infinitely many $n$, let $s$ be minimal for $a$'s preference, and let $M$ be large enough such that every profile $s_n$ with $M < n$ occurs infinitely often in the sequence. Let $s = s_n$ for some $n > M$, and let $k > n$ be the least $k$ such that $s_n$ and $s_k$ induce the same play. Lemma~\ref{lem:avoid-play} implies that $s_{k-1} \rightharpoonup_a s_k$, so player $a$ prefers the outcome of $s_n$ over that of $s_{k-1}$, contradiction.
\end{proof}

Together with Corollary~\ref{corr:nash-lazy-term} the following corollary shows the equivalence between all preferences being acyclic and universal existence of NE.

\begin{corollary}\label{cor:all-acycl}
Consider outcomes $O$, players $A$, and their preferences $(\prec_a)_{a \in A}$: All $\prec_a$ are acyclic iff for all finite sequential games built from $O$, $A$ and $(\prec_a)_{a \in A}$ the collective lazy improvement terminates.
\end{corollary}

\begin{proof}
The difficult implication of the equivalence is a corollary of Theorem~\ref{thm:lazy-term}. For the other implication, note that if $x_0\prec_a x_1\prec_a\dots\prec_a x_n\prec_a x_0$, then $\rightharpoonup_a$ does not terminate on the profile below.

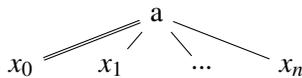
\begin{tikzpicture}[level distance=7mm]
\node{a}[sibling distance=12mm]
	child{node{$x_0$}[sibling distance=10mm]edge from parent[double]}
	child{node{$x_1$}[sibling distance=10mm]}
	child{node{...}[sibling distance=10mm]}
	child{node{$x_n$}[sibling distance=10mm]}
	;
\end{tikzpicture}
\end{proof}

\begin{corollary}\label{cor:acyclic-semi-pot}
In a finite sequential game where every player has acyclic preferences, lazy improvement is a semi-potential.
\begin{proof}
Combine Corollary \ref{corr:nash-lazy-term} and Corollary \ref{cor:all-acycl}.
\end{proof}
\end{corollary}

\name{Kukushkin} (\cite[Theorem 3]{kukushkin2}) proved Corollary~\ref{cor:acyclic-semi-pot} in the case where the preferences are derived from payoffs. In this specific (yet usual) setting, it is not possible to consider players with cyclic preferences, so Theorem~\ref{thm:lazy-term} or Corollary~\ref{cor:all-acycl} cannot even be stated.

Based on Corollary \ref{cor:all-acycl} we obtain a reasonable candidate for rational behaviour in sequential games played with an unpredictable nature or erratic players: Perform lazy improvement until the players with acyclic preferences no longer change their strategies. It is always consistent with the observations to assume that the changes in another player's strategy are based on lazy convertibility. This argument is explored in more detail in Section \ref{sec:beliefs}. Nature can then be modelled as a player with the full relation as preferences, such that any convertible step for nature becomes an improvement step.

\section{Termination in finite games, second proof}\label{sect:finite-2nd-proof}

The proof of Theorem~\ref{thm:lazy-term}, by contradiction, gives a quick argument but no deep insight on how and how fast the relation terminates. A stronger statement can proven by using the multiset of outcomes avoided by a player $a$ (i.e.~the outcomes obtained in a subgame, where the decision not to play into that subgame was made by $a$, see Definition~\ref{defn:sgdo}) to construct a measure that will decrease on any lazy improvement step by $a$ (Lemma~\ref{lem:lazy-diff}), and remain unchanged by any lazy convertibility step by a different player (Lemma~\ref{lem:lazy-same}). Thus, we are in a situation very similar to potential games \cite{shapley} -- however, in a potential game a player \emph{can} increase the potential (which is common to all the players) but does not \emph{want} to, whereas here the players \emph{cannot} impact the measure of another player as long as they are restricted to lazy convertibility.

\begin{definition}[Avoided outcomes of a game and of a profile]\label{defn:sgdo}
The avoided outcomes of a game $g$ is a function $\Delta(g)$ of type $A\to\mathbb{N}$, and it is defined inductively below.
\begin{itemize}
\item $\Delta(g,a):=0$ if $g$ is a leaf game.
\item If player $a$ owns the root of a game $g$ whose children are $g_0,\dots,g_n$ then
\begin{itemize}
\item $\Delta(g,b):=\sum_{j=0}^n\Delta(g_j,b)$ for all $b\neq a$.
\item $\Delta(g,a):=\big(\sum_{j=0}^n\Delta(g_j,a)\big)+n$
\end{itemize}
\end{itemize}

The avoided outcomes of a profile $s$ is a function $\delta(s)$ of type $A\to O\to\mathbb{N}$, or equivalently in this case, of type $A\times O\to\mathbb{N}$, and it is defined inductively below.
\begin{itemize}
\item $\delta(s,a,o):=0$ if $s$ is a leaf profile.
\item If player $a$ owns the root of a profile $s$ and chooses the subprofile $s_i$ among $s_0,\dots,s_n$ then
\begin{itemize}
\item $\delta(s,b,o):=\sum_{j=0}^n\delta(s_j,b,o)$ for all $b\neq a$.
\item $\delta(s,a,o):=\big(\sum_{j=0}^n\delta(s_j,a,o)\big)+|\{j\in\{0,\dots,n\}-\{i\}\,\mid\,v(s_j)=o\}|$
\end{itemize}
\end{itemize}
\end{definition}

The smaller array below describes the function $\Delta(g)$, where $g$ is the underlying game of the left-hand profile $s$ below, and the right-hand array describes the function $\delta(s)$. For instance $\delta(s,b,y)=2$ because player $b$ avoids the outcome $y$ twice: once at the left-most internal node, after two leftward moves, when choosing outcome $x$ rather than $y$, and also once after one rightward move, also when choosing $x$ rather than $y$. Note that the only leaf that is not accounted for by the function of the avoided outcome of a profile/game is the leaf that is induced by the profile.

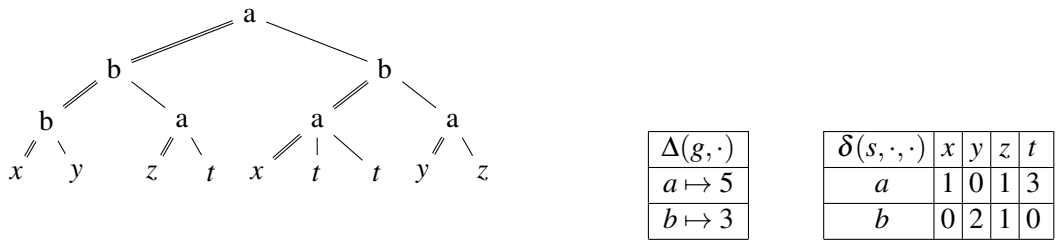
\begin{displaymath}
\begin{array}{c@{\hspace{2cm}}c@{\hspace{1cm}}c}
\begin{tikzpicture}[level distance=7mm]
\node{a}[sibling distance=36mm]
	child{node{b}[sibling distance=18mm] edge from parent[double]
		child{node{b}[sibling distance=8mm]edge from parent[double]
			child{node{$x$}edge from parent[double]}
			child{node{$y$}}
		}
		child{node{a}[sibling distance=8mm]
			child{node{$z$}edge from parent[double]}
			child{node{$t$}}
		}
	}
	child{node{b}[sibling distance=18mm]
		child{node{a}[sibling distance=8mm]edge from parent[double]
			child{node{$x$}edge from parent[double]}
			child{node{$t$}}
			child{node{$t$}}
		}
		child{node{a}[sibling distance=8mm]
			child{node{$y$}edge from parent[double]}
			child{node{$z$}}
		}
	}
	;
\end{tikzpicture}
&
\begin{array}{|c|}
	\cline{1-1}
	\Delta(g,\cdot)\\
	\cline{1-1}
	a \mapsto 5\\
	\cline{1-1}
	b \mapsto 3\\
	\cline{1-1}
\end{array}
&
\begin{array}{|c@{\;\vline\;}c@{\;\vline\;}c@{\;\vline\;}c@{\;\vline\;}c|}
	\cline{1-5}
	\delta(s,\cdot,\cdot) & x & y & z & t\\
	\cline{1-5}
	a & 1 & 0 & 1 & 3\\
	\cline{1-5}
	b & 0 & 2 & 1 & 0\\
	\cline{1-5}
\end{array}
\end{array}
\end{displaymath}

Observation~\ref{obs:sgdo} below relates the two functions from Definition~\ref{defn:sgdo}. It refers to $s2g$, a function that returns the underlying game of a given profile, see \cite{SLR-PhD08} or \cite{SLR09} for a proper definition.

\begin{observation}\label{obs:sgdo}
\begin{enumerate}
\item\label{obs:sgdo1} Let $s$ be a profile and $a$ be a player, then $\Delta(s2g(s),a)=\sum_{o\in O}\delta(s,a,o)$.
\item\label{obs:sgdo2} Let $g$ be a game, then $1+\sum_{a\in A}\Delta(g,a)$ equals the number of leaves of $g$.
\end{enumerate}
\end{observation}

\begin{proof}
\begin{enumerate}
\item By induction on $s$. If $s$ is a leaf profile, the claim holds since $\Delta(s2g(s),a)=0=\delta(s,a,o)$ by definition, so now let $s$ be a profile where the root owner $a$ chooses $s_i$ among subprofiles $s_0,\dots,s_n$. For $b\neq a$ Definition~\ref{defn:sgdo} and the induction hypothesis yield $\Delta(s2g(s),b)=\sum_{j=0}^n\Delta(s2g(s_j),b)\stackrel{I.H.}{=}\sum_{j=0}^n\sum_{0\in O}\delta(s_j,b,o)=\sum_{0\in O}\sum_{j=0}^n\delta(s_j,b,o)=\sum_{0\in O}\delta(s,b,o)$. Similarly we have $\Delta(s2g(s),a) = \sum_{j=0}^n \Delta(s2g(s_j),a)+n\stackrel{I.H.}{=}\sum_{j=0}^n\sum_{o\in O}\delta(s_j,a,o)+|\{j\in\{0,\dots,n\}-\{i\}\,\mid\,v(s_j)\in O\}| =\\\sum_{o\in O} \big(\sum_{j=0}^n \delta(s_j,a,o)+|\{j\in\{0,\dots,n\}-\{i\}\,\mid\,v(s_j)=o\}|\big)=\sum_{o\in O}\delta(s,a,o)$.

\item By induction on $g$. This holds for every leaf game $g$ since $\Delta(g,a)=0$ by definition. Let $g$ be a game whose root is owned by player $a$ and whose subgames are $g_0,\dots,g_n$. The number of leaves in $g$ is the sum of the numbers of leaves in the $g_j$, that is, $\sum_{j=0}^{n}\big(1+\sum_{b\in A}\Delta(g_j,b)\big)$ by induction hypothesis. This, equals $1+\sum_{j=0}^{n}\sum_{b\in A-\{a\}}\Delta(g_j,b)+ n + \sum_{j=0}^{n}\Delta(g_j,a)$, which, in turn, equals $1+\sum_{b\in A-\{a\}}\Delta(g,b)+\Delta(g,a)$ by definition.
\end{enumerate}
\end{proof}

Lemma~\ref{lem:lazy-same} below states conservation of the outcomes that are avoided by a player in a profile during a lazy conversion of another player. Intuitively, it is because a lazy conversion of a player cannot modify the subtrees that are avoidd by the other players, even though she owns node therein.

\begin{lemma}\label{lem:lazy-same}
$s\stackrel{c}{\rightharpoonup}_as'\,\wedge\,b\neq a\quad\Rightarrow\quad\delta(s,b)=\delta(s',b)$
\end{lemma}

\begin{proof}
By induction on the profile. It holds for leaves, so let $s\stackrel{c}{\rightharpoonup}_as'$ with subprofiles $s_0,\dots,s_n$ and $s'_0,\dots,s'_n$, respectively. By definition of $\stackrel{c}{\rightharpoonup}_a$ we have $s_j\stackrel{c}{\rightharpoonup}_as'_j$ for all $j$, and therefore $\delta(s_j,b,o)=\delta(s'_j,b,o)$ by induction hypothesis. If the root owner is different from $b$, then $\delta(s,b,o)=\sum_{j=0}^n\delta(s_j,b,o)=\sum_{j=0}^n\delta(s'_j,b,o)=\delta(s',b,o)$ by definition of $\delta$. If $b$ is the root owner, she chooses the $i$-th subprofile in both $s$ and $s'$ since $b\neq a$, and moreover $s'_j=s_j$ for all $j$ distinct from $i$. So $\delta(s,b,o)=\sum_{j=0}^n\delta(s_j,b,o)+|\{j\in\{0,\dots,n\}-\{i\}\,\mid\,v(s_j)=o\}|=\sum_{j=0}^n\delta(s'_j,b,o)+|\{j\in\{0,\dots,n\}-\{i\}\,\mid\,v(s'_j)=o\}|=\delta(s',b,o)$.
\end{proof}

However, the conservation does not fully hold for the player who converts the profile, unless the induced outcomes are the same for both profiles. The difference is little though, only depending on both induced outcomes. In Lemma~\ref{lem:lazy-diff} below, $ eq$ is just a boolean representation of equality: $ eq(x,x):=1$ and $ eq(x,y):=0$ for $x\neq y$.

\begin{lemma}\label{lem:lazy-diff}
$s\stackrel{c}{\rightharpoonup}_as'\quad\Rightarrow\quad\delta(s,a)+ eq(v(s))=\delta(s',a)+ eq(v(s'))$
\end{lemma}

\begin{proof}
By induction on the profile $s$. It holds for leaves, so let $s\stackrel{c}{\rightharpoonup}_as'$ with subprofiles $s_0,\dots,s_n$ and $s'_0,\dots,s'_n$, respectively. If the root owner is distinct from $a$, she chooses the same $i$-th subprofile in both $s$ and $s'$, therefore $\delta(s,a,o)+ eq(v(s),o)=\sum_{0\leq j\leq n\,\wedge\,j\neq i}\delta(s_j,a,o)+\delta(s_i,a,o)+ eq(v(s_i),o)=\sum_{0\leq j\leq n\,\wedge\,j\neq i}\delta(s'_j,a,o) + \delta(s'_i,a,o) + eq(v(s'_i),o)=\delta(s',a,o)+ eq(v(s'),o)$ by definition of $\delta$, since $s_j=s'_j$ for $j\neq i$, and by induction hypothesis.

If $a$ is the root owner, let $a$ choose the $i$-th and $k$-th subprofiles in $s$ and $s'$, respectively. Let $N:=\delta(s,a,o)+ eq(v(s),o)$, so $N=\sum_{0\leq j\leq n\,\wedge\,j\neq k}^n\delta(s'_j,a,o)+|\{j\in\{0,\dots,n\}-\{i\}\,\mid\,v(s_j)=o\}|+\delta(s_k,a,o)+ eq(v(s_i),o)$ by unfolding Definition~\ref{defn:sgdo}, since $s'_j=s_j$ for all $j\neq k$, and since $v(s)=v(s_i)$ by the choice at the root. Rewriting $N$ twice with the easy-to-check equality $|\{j\in\{0,\dots,n\}-\{x\}\,\mid\,v(s_j)=o\}|+ eq(v(s_x),o)=|\{j\in\{0,\dots,n\}\,\mid\,v(s_j)=o\}|$, first with $x:=i$ and then with $x:=k$ yields the equality $N=\sum_{0\leq j\leq n\,\wedge\,j\neq k}^n\delta(s'_j,a,o)+|\{j\in\{0,\dots,n\}-\{k\}\,\mid\,v(s_j)=o\}|+\delta(s_k,a,o)+ eq(v(s_k),o)$. Since $s_k\stackrel{c}{\rightharpoonup}_as'_k$ by definition of lazy convertibility, and by the induction hypothesis, let us further rewrite $\delta(s_k,a)+ eq(v(s_k))$ with $\delta(s'_k,a)+ eq(v(s'_k))$ in $N$. Folding Definition~\ref{defn:sgdo} yields $N=\delta(s',a,o)+ eq(v(s'),o)$.
\end{proof}

The two lemmas above suggest that whenever a player lazily converts a profile to obtain a better outcome, some measure decreases a bit with respect to her preference, but does not change for the other players. The lazy improvement should therefore terminate, and even quite quickly, as proved below. Recall that a finite preference relation $\prec $ has height at most $h$ if there is no chain $s_1 \prec s_2 \prec \ldots \prec s_{h+1}$.

\begin{theorem}\label{thm:lazy-term2}
Consider a game $g$ where player $a$ has an acyclic preference of height $h$. Let $\Delta(g,a)$ be the total number of choices available to player $a$, minus the number of vertices where $a$ is choosing. Then in any sequence (possibly infinite) of lazy improvement, the number of lazy improvement steps performed by player $a$ is bounded by $(h-1)\cdot\Delta(g,a)$.

\begin{proof}
For every outcome $o$ let $h(a,o)$ be the maximal cardinality of the $\prec_a$-chains whose $\prec_a$-maximum is $o$, and note that $o\prec_ao'$ implies $h(a,o)<h(a,o')$. For every profile $s$ let $M(s,a):=\sum_{o\in O}(h(a,o)-1)\cdot\delta(s,a,o)$ and note that $0\leq M(s,a)\leq (h-1)\cdot \Delta(g,a)$ by Observation~\ref{obs:sgdo}.\ref{obs:sgdo1}. Let $s\rightharpoonup_as'$ be a lazy improvement step, so $s\stackrel{c}{\rightharpoonup}_as'$ and $v(s)\prec_av(s')$ by definition, then $M(s,a)-M(s',a)=\sum_{o\in O}(h(a,o)-1)\cdot(\delta(s,a,o)-\delta(s',a,o))=h(a,v(s'))-h(a,v(s))>0$ by Lemma~\ref{lem:lazy-diff}. Let $s \rightsquigarrow_a s'$ be a lazy equilibrium perturbation step. By  Lemma~\ref{lem:lazy-diff} and the definition of $s \sim s'$, we have $M(s,a) = M(s',a)$. Let $s\stackrel{c}{\rightharpoonup_b}s'$ be a lazy conversion step where $b\neq a$, then $M(s,a)=M(s',a)$ by Lemma~\ref{lem:lazy-same}. This shows that the $\rightharpoonup_a$ steps are at most $(h-1)\cdot\Delta(g,a)$ in any sequence of $\rightharpoonup \cup \rightsquigarrow$.
\end{proof}
\end{theorem}

\begin{corollary}\label{cor:all-acycl2}
The lazy improvement terminates for all games iff all preferences are acyclic, in which case the number of sequential lazy improvement steps is at most $(h-1)\cdot (l-1)$ where $h$ bounds the cardinality of the preference chains and $l$ is the number of leaves.
\end{corollary}

\begin{observation}\label{obs:quadratic}
\begin{enumerate}
\item The maximal length of a lazy improvement sequence is bounded in a quadratic manner in the size of the game in general and linearly when $h$ from Corollary~\ref{cor:all-acycl2} is fixed.
\item\label{rem-quad2} The quadratic and linear bounds are tight.
\end{enumerate}
\end{observation}

\begin{proof}(of \ref{obs:quadratic}.\ref{rem-quad2}.)
For the linear bound, let us consider the figure below and set $x := x_0 = \dots = x_n$ and $y \prec_a x$ and $x \prec_b y$. There is clearly a lazy improvement sequence starting from the figure and visiting each leaf exactly once.

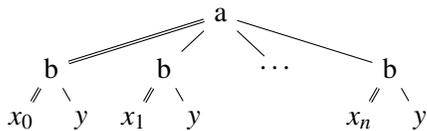
\begin{tikzpicture}[level distance=7mm]
\node{a}[sibling distance=15mm]
	child{node{b}[sibling distance=8mm] edge from parent[double]
		child{node{$x_0$} edge from parent[double]}
		child{node{$y$}}
	}
	child{node{b}[sibling distance=8mm]
		child{node{$x_1$} edge from parent[double]}
		child{node{$y$}}
	}
	child{node{\dots}}
	child{node{b}[sibling distance=8mm]
		child{node{$x_n$} edge from parent[double]}
		child{node{$y$}}
	};
\end{tikzpicture}

It is similar for the quadratic bound, but we need to be a bit more careful. For $n\in\mathbb{N}$, consider the game in the above figure, where $y\prec_ax_0\prec_ax_1\dots\prec_ax_n$ and $x_i\prec_by$ for all $i$. Let us prove by induction on $n$ the existence of a sequence of $\frac{(n+2)(n+3)}{2}-2$ lazy improvement steps when starting from the strategy profile above. For the base case $n=0$, there are $1=\frac{(0+2)(0+3)}{2}-2$ lazy improvement steps. For the inductive case, let player $a$ make $n$ lazy improvements in a row, by choosing $x_1$, then $x_2$, and so on until $x_n$. At that point, let player $b$ improve from $x_n$ to $y$ and then let player $a$ come back to $x_0$. So far, $n+2$ lazy improvement steps have been performed. Now let us ignore the substrategy profile involving $x_n$ (and $y$). By induction hypothesis, $\frac{(n+1)(n+2)}{2}-2$ additional lazy improvement steps can be performed in a row. Since $(n+2)+\frac{(n+1)(n+2)}{2}-2=\frac{(n+2)(n+3)}{2}-2$, we are done.
\end{proof}

\section{Lazy convertibility as belief updating}
\label{sec:beliefs}
Let us discuss whether we should expect players to conform to lazy convertibility when playing a sequential game repeatedly. Observation \ref{obs:outcomes} tells us that in the short term, a player has no incentive to deviate from lazy convertibility: If she desires some outcome she can reach by some deviation from her current strategy, she can obtain this outcome by converting a strategy in a lazy way. There is a caveat, though, in that restricting convertibility to lazy convertibility changes the overall reachability structure, as the following example shows.

\begin{example}
The last profile of the three-step improvement relation below is a Nash equilibrium that cannot be reached from the first profile under lazy improvement.

{\small
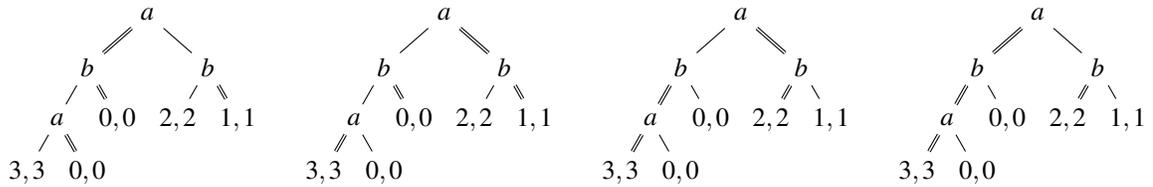
\begin{tabular}{ccccc}
\begin{tikzpicture}[level distance=7mm]
\node{a}[sibling distance=16mm]
	child{node{b}[sibling distance=8mm]edge from parent[double]
		child{node{a}[sibling distance=8mm]
			child{node{$3,3$}}
			child{node{$0,0$}edge from parent[double]}
		}
		child{node{$0,0$}edge from parent[double]}
	}
	child{node{b}[sibling distance=8mm]
			child{node{$2,2$}}
			child{node{$1,1$}edge from parent[double]}
	};
\end{tikzpicture}
&
\begin{tikzpicture}[level distance=7mm]
\node{a}[sibling distance=16mm]
	child{node{b}[sibling distance=8mm]
		child{node{a}[sibling distance=8mm]
			child{node{$3,3$}edge from parent[double]}
			child{node{$0,0$}}
		}
		child{node{$0,0$}edge from parent[double]}
	}
	child{node{b}[sibling distance=8mm]edge from parent[double]
			child{node{$2,2$}}
			child{node{$1,1$}edge from parent[double]}
	};
\end{tikzpicture}
&
\begin{tikzpicture}[level distance=7mm]
\node{a}[sibling distance=16mm]
	child{node{b}[sibling distance=8mm]
		child{node{a}[sibling distance=8mm]edge from parent[double]
			child{node{$3,3$}edge from parent[double]}
			child{node{$0,0$}}
		}
		child{node{$0,0$}}
	}
	child{node{b}[sibling distance=8mm]edge from parent[double]
			child{node{$2,2$}edge from parent[double]}
			child{node{$1,1$}}
	};
\end{tikzpicture}
&
\begin{tikzpicture}[level distance=7mm]
\node{a}[sibling distance=16mm]
	child{node{b}[sibling distance=8mm]edge from parent[double]
		child{node{a}[sibling distance=8mm]edge from parent[double]
			child{node{$3,3$}edge from parent[double]}
			child{node{$0,0$}}
		}
		child{node{$0,0$}}
	}
	child{node{b}[sibling distance=8mm]
			child{node{$2,2$}edge from parent[double]}
			child{node{$1,1$}}
	};
\end{tikzpicture}
\end{tabular}
}
\end{example}

From the perspective of any given player, it however makes a lot of sense to assume that all other players are updating their own strategies only in a lazy way -- assuming that only relevant choices of the other players can be observed. The latter seems to be crucial in order to make the game truly sequential: If all players announced their entire strategy simultaneously, it would be a game in normal form after all.

To formalize this idea, let us fix a player $a$ and consider the game from her perspective. She may consider the game as a two-player game played by her against all other players aggregated into a single player $b$. She starts with some initial strategy $s_a^{(0)}$, and some prior assumption $s_b^{(0)}$ on the strategy of her opponent(s). She then updates her own strategy via lazy improvement to $s_a^{(1)}$. Then the game is actually played, and player $a$ observes the actual moves (but not the strategy) of her opponents. As she only observes the moves along the path actually taken, it is consistent with her observations to assume that the aggregated opponent player lazily converted $s_b^{(0)}$ into some $s_b^{(1)}$. Then the player $a$ again performs a lazy improvement step to $s_a^{(2)}$, plays the game, etc. Provided that the player $a$ has acyclic preferences, Theorem \ref{thm:lazy-term} implies that her own strategy stabilizes to some strategy $s_a$ eventually.

Note that this procedure requires no assumptions on knowledge of rationality of other players or their payoff functions, not to speak of common knowledge. There is in general no reason to assume that the aggregated player $b$ acts according to acyclic preferences (given that the different players making up $b$ may have partially antagonistic preferences). However, if each player has a cyclic preference and performs the same procedure as $a$ above, then each players actual strategy will stabilize. As any change in what a player assumes her aggregated opponents are playing has to be caused by either a change in her own, or someone else's strategy, this implies that also the believed strategy of the aggregated players $s_b$ will stabilize. Furthermore, all the strategy profiles constructed in this way induce the same play, and combining them as follows yields a Nash equilibrium:

\begin{proposition}
Let a set of players play a finite sequential game by converting their own strategies lazily based on beliefs about the other players strategies in order to maximize an acyclic preference relation. Then a Nash equilibrium can be obtained from the stable strategies they will settle to as follows: Along the common path chosen by their stable strategies, everyone follows their own strategy. In any subgame that is not reached, each player plays according to the beliefs held by the player controlling access to the subgame about their strategies.
\begin{proof}
At any vertex reached during the final play, the choice facing the current player is the same one she was anticipating due to her beliefs on her opponents strategies. As her choice is consistent with the stable choice made during the dynamical updating, she has no incentive to change.
\end{proof}
\end{proposition}

In comparison, the investigation of the epistemic foundations of Nash equilibria by \name{Aumann} and \name{Brandenburger} \cite{aumann5} identified mutual knowledge of rationality, knowledge of the game and (in case of more than two players) a common prior as the prerequisite for playing a Nash equilibrium. A subgame perfect equilibrium requires even stronger assumptions, namely well-aware players \cite{aumann4}.

\section{Convergence in infinite games}
\label{sec:infinite}
Infinite sequential games with win/lose preferences are generalizations of Gale-Stewart games \cite{gale2}, and are of great relevance for logic. That any two-player game with antagonistic preferences and a Borel winning set actually has a Nash equilibrium is a highly non-trivial result by \name{Martin} \cite{martin}. This is generalized to multi-player games with non-necessarily antagonistic preferences in \cite{leroux3}. Moreover, subgame-perfect equilibria are not always guaranteed to exist (cf.~\cite{paulyleroux2,solan}).

The definition of lazy improvement applies to infinite sequential games as well, and we can adapt the results on finite games to see that it still constitutes a semi-potential:

\begin{proposition}
Consider an infinite sequential game where each player (there might be countably many) has acyclic preferences. Then lazy improvement is a semi-potential.
\begin{proof}
As argued in Section \ref{sec:definition}, we only need to show that lazy improvement is acyclic. Assume the contrary, then there is some finite cycle $s_1 \rightharpoonup s_2 \rightharpoonup \ldots \rightharpoonup s_n \rightharpoonup s_1$. Let $p_1,\ldots,p_n$ be the paths induced by these strategy profiles, and choose $k \in \mathbb{N}$ such that $p_i|_{\leq k} = p_j|_{\leq k} \Leftrightarrow p_i = p_j$.

By choice of $k$, the path chosen inside any subgame rooted at depth $k$ remains unchanged throughout the improvement cycle. Thus, replacing any such subgame with a leaf carrying the outcome induced by this path has no impact on the improvement cycle. We have obtained a finite sequential game with the same preferences and a cycle built from lazy improvement step, contradicting Theorem \ref{thm:lazy-term}.
\end{proof}
\end{proposition}

Of course, in an infinite game acyclicity does not suffice to ensure termination or even convergence. In fact, \cite[Example 26]{paulyleroux2} (reproduced below as Example \ref{ex:twentysix}) shows that lazy improvement in infinite games will not always converge, and that even accumulation points do not have to be Nash equilibria. There are however several potential ways to extend the results on lazy improvement to infinite sequential games:

\begin{enumerate}
\item We can consider games where the preferences are expressed via continuous payoff functions. For some fixed $\varepsilon > 0$, we can then consider $\varepsilon$-lazy improvement (where only lazy convertibility is allowed, and improvement steps are only taken if the player can improve by more than $\varepsilon$). Then Theorems \ref{thm:lazy-term} and \ref{thm:lazy-term2} carry over, and as a counterpart to Corollary \ref{corr:nash-lazy-term} we find that the terminal profiles of $\varepsilon$-lazy improvement are precisely the $\varepsilon$-Nash equilibria. See \cite{leroux2} for details.
\item We can employ lazy improvement being done in a finitary way with increasing precision, and find that any accumulation point of particular subsequence is guaranteed to be a Nash equilibrium. This was discussed in \cite[Section VI]{paulyleroux2}.
\item We can consider games where the players have win/lose objectives (i.e.~their preference relations have height $2$), and the winning sets are $\Delta^0_2$-sets. Then transfinite iteration of lazy improvement will reach a Nash equilibrium.
\item We can generalize the measure employed in the proof of Theorem \ref{thm:lazy-term2} to infinite games with Lipschitz payoff functions.
\end{enumerate}

The rest of the section provides some details on the third and fourth approaches. We start by formalizing what it means to do a transfinite number of improvement steps. The following definition generalises the notion of finite sequence or $\omega$-sequence induced by a binary relation to $\alpha$-sequence for some ordinal $\alpha$: at limit ordinals, following a valid sequence amounts to picking an "accumulation point".

\begin{definition}[ordinal sequence of a relation]
Let $\to$ be a binary relation on some topological space $S$, and let $\alpha$ be an ordinal number. An $\alpha$-sequence of $\to$ is a family $(s_\beta)_{\beta<\alpha}$ of elements in $S$ such that for all $\beta < \alpha$, if $\beta +1 < \alpha$ then $s_{\beta}\to s_{\beta +1}$, and if $\beta$ is a limit ordinal, then for every $\beta' < \beta$ and every neighbourhood $U$ of $s_\beta$ there exists $\gamma\in]\beta',\beta[$ such that $s_{\gamma}\in U$.
\end{definition}

Observation~\ref{obs:no-def-sink} below says that, given a binary relation over a sequentially compact set, the only reason why an ordinal sequence cannot be further extended is when a sink has been reached.

\begin{observation}\label{obs:no-def-sink}
Let  $(s_\beta)_{\beta<\alpha}$ be a countable ordinal sequence of $\to$ over a sequentially compact set $S$. If $(s_\beta)_{\beta \leq \alpha}$ is not a sequence of $\to$ for any $s_{\alpha}\in S$, then $\alpha = \alpha'+1$ for some $\alpha'$ and $s_{\alpha'}$ is a sink of $\to$.
\end{observation}

In this paper by countable we mean at most countable. We do find that even in very simple games, we can have improvement sequences of any countable length:
\begin{proposition}
\label{prop:alphasequence}
For every countable ordinal $\alpha$ there exists a win-lose two-player game on a binary tree with open winning set for one player, and an $\alpha$-sequence of lazy improvement in the game.
\end{proposition}

Lemma~\ref{lem:objectives-open-closed} below is the base case of the proof of Theorem~\ref{thm:lazy-dh}, which is proved by transfinite induction in the difference hierarchy.

\begin{lemma}\label{lem:objectives-open-closed}
Let $g$ be a game with finitely many players who have win-lose objectives. If every winning set is open or closed, every sequence of lazy improvement in $g$ is countable.
\end{lemma}

Recall that a subset $S$ of a metric space is called a $\Delta^0_2$-set, if it is expressible both as $S = \bigcap_{i \in \mathbb{N}} U_i$ with open $U_i$, and as $S = \bigcup_{i \in \mathbb{N}} A_i$ with closed $A_i$.

\begin{theorem}\label{thm:lazy-dh}
Let $g$ be a game with finitely many players who have win-lose objectives. If every winning set is $\Delta^0_2$, every sequence of lazy improvement in $g$ is countable.
\end{theorem}

\begin{corollary}\label{cor:lazy-dh}
Let $g$ be a game with finite branching and finitely many players who have win-lose objectives. If every winning set is $\Delta^0_2$, every sequence of lazy improvement in $g$ is countable and ends at a Nash equilibrium.

\begin{proof}
By Theorem~\ref{thm:lazy-dh} and Observation~\ref{obs:no-def-sink}, since finite branching implies compactness.
\end{proof}
\end{corollary}

Regarding a potential extension of Corollary \ref{cor:lazy-dh} to winning sets beyond $\Delta^0_2$ we shall make a tangential remark: The computational task of finding a Nash equilibrium in a two-player sequential game with $\Delta^0_2$ winning sets is just as hard as iterating the task of finding an accumulation point of a sequence over some countable ordinal. This follows from results in \cite{paulyleroux3-cie,paulyleroux3-arxiv,pauly-ordinals,gherardi4}. Finding a Nash equilibrium of a game with $\Sigma^0_2$ winning sets is strictly more complicated. Thus, $\Delta^0_2$ seems to be a natural boundary for results of the form of Corollary \ref{cor:lazy-dh}.

We shall now proceed to discuss the fourth approach, based on Lipschitz payoff functions.

\begin{definition}[Fair improvement\footnote{This is \emph{fair} as in \emph{fair scheduler}, not as in \emph{fair division of cake}.}]
An infinite sequence of improvements is fair if the following holds: for all positive real numbers $r$, if improvements by more than $r$ are possible infinitely often during the sequence, they also occur infinitely often.
\end{definition}


\begin{proposition}\label{prop:Lipschitz-fair-lazy}
If the game tree is binary and if for each player there exists $\eta > 2$ such that her payoff function is Lipschitz-continuous for the distance $d$ defined by $d(h0\rho,h1\rho') = \frac{1}{\eta^{|h|}}$, then all the accumulation points of a fair lazy improvement sequence are Nash equilibria.

\begin{proof}
To all strategy profiles $s$ and all players $a$ let us associate a real number:

 \[M_a(s) := \sum_{h\in d^{-1}(a)} v_a \big( h\cdot (1-s(h)) \cdot \rho(h\cdot (1-s(h)),s) \big) - \min_{\rho\in \{0,1\}^\omega}(v_a(h\rho))\]

\noindent where $d(h)$ is the player that plays at history $h$, and $v_a(\rho)$ is the payoff for player $a$ and run $\rho$, and $s(h)$ is the choice in $\{0,1\}$ that is prescribed by $s$ at $h$, and $\rho(h,s)$ is the run induced by strategy profile $s$ from $h$ on. Similarly to the finite case $v_a( h\cdot (1-s(h)) \cdot \rho(h\cdot(1-s(h)),s))$ is the payoff that is avoided by $a$ at history $h$. Note that the summands of $M_a(s)$ are all non-negative by definition of the minimum, and that the sum converges absolutely: indeed, by assumption $|v_a(h0\rho)-v_a(h1\rho')| \leq \frac{L_a}{\eta^{|h|}}$ for some $L_a > 0$ and for all $h$, $\rho$, and $\rho'$, so $M_a(s) \leq \sum_{h\in \{0,1\}^*} \frac{L_a}{\eta^|h|} = L_a \sum_{l = 0}^{+\infty}(\frac{2}{\eta})^l = \frac{L_a}{1-\frac{2}{\eta}}$. Also, each $M_a$ is continuous.

Similarly to the finite case, $M_a$ decreases by $x$ when player $a$ performs a lazy improvement by $x$, and $M_a$ is left unchanged when another player performs a lazy convertibility step. Let $(s_n)_{n \in \mathbb{N}}$ be a fair lazy improvements sequence,
and let $s$ be some accumulation point that is not an Nash equilibrium. So $s \rightharpoonup_a t$ for some profile $t$ and player $a$. By continuity of the payoffs we can assume wlog that $s$ and $t$ only disagree on a finite set $H_0$ of histories shorter than some $k$. Let $r := v_a(t) - v_a(s) > 0$ and $\delta > 0$ be such that $|v_a(\sigma) - v_a(\sigma') | < \frac{r}{4}$ whenever $d'(\sigma,\sigma') < \delta$ with $d'(\sigma,\sigma') := \frac{1}{2^i}$ where $i$ is maximal such that $\sigma$ and $\sigma'$ coincide for histories not longer than $i$. By the Bolzano-Weierstrass theorem there is a subsequence $(s_{\varphi(n)})$ converging towards $s$, and $d'(s,s_{\varphi(n)}) < \min(\delta, \frac{1}{2^k})$ for all but finitely many $n$. When the inequality holds, let $t_n$ coincide with $s_{\varphi(n)}$ outside of $H_0$ and with $t$ on $H_0$. So $s_{\varphi(n)}  \stackrel{c}{\rightharpoonup}_a t_n$ and $d'(t,t_n) = d'(s,s_{\varphi(n)})$, so $|v_a(t_n) - v_a(t)| < \frac{r}{4}$, and $ \frac{r}{2} < v_a(t_n) - v_a(s_{\varphi(n)})$. Thus improvement by at least $\frac{r}{2}$ are infinitely often possible for player $a$ in the sequence. By fairness infinitely many of them must occur, so $M_a$ diverges to $-\infty$, contradiction.
\end{proof}
\end{proposition}

\section{Some counterexamples}
\label{sec:counter}
In order to obtain the termination result in the finite case (Theorem \ref{thm:lazy-term}), some restriction on how players can improve is indeed necessary. We shall show below that the collective improvement $\twoheadrightarrow$ may fail to terminate even for very simple games in extensive form:

\begin{example}
An improvement cycle:\\
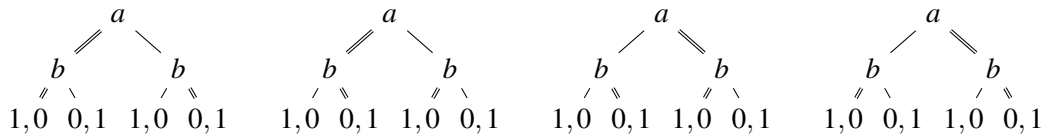
\begin{tabular}{cccc}
\begin{tikzpicture}[level distance=7mm]
\node{a}[sibling distance=16mm]
	child{node{b}[sibling distance=8mm] edge from parent[double]
		child{node{$1,0$} edge from parent[double]}
		child{node{$0,1$}}
	}
	child{node{b}[sibling distance=8mm]
			child{node{$1,0$}}
			child{node{$0,1$} edge from parent[double]}
	};
\end{tikzpicture}
&
\begin{tikzpicture}[level distance=7mm]
\node{a}[sibling distance=16mm]
	child{node{b}[sibling distance=8mm] edge from parent[double]
		child{node{$1,0$}}
		child{node{$0,1$} edge from parent[double]}
	}
	child{node{b}[sibling distance=8mm]
			child{node{$1,0$} edge from parent[double]}
			child{node{$0,1$}}
	};
\end{tikzpicture}
&
\begin{tikzpicture}[level distance=7mm]
\node{a}[sibling distance=16mm]
	child{node{b}[sibling distance=8mm]
		child{node{$1,0$}}
		child{node{$0,1$} edge from parent[double]}
	}
	child{node{b}[sibling distance=8mm] edge from parent[double]
			child{node{$1,0$} edge from parent[double]}
			child{node{$0,1$}}
	};
\end{tikzpicture}
&
\begin{tikzpicture}[level distance=7mm]
\node{a}[sibling distance=16mm]
	child{node{b}[sibling distance=8mm]
		child{node{$1,0$} edge from parent[double]}
		child{node{$0,1$}}
	}
	child{node{b}[sibling distance=8mm] edge from parent[double]
			child{node{$1,0$}}
			child{node{$0,1$} edge from parent[double]}
	};
\end{tikzpicture}
\end{tabular}
\end{example}

The technical notion of strategy that is used in this article to represent the intuitive concept of a strategy (in games in extensive form) is not the only possible notion. An alternative notion does not require choices from a player at every node that she owns, but only at nodes that are not ruled out by the strategy of the same player. The three objects in Example \ref{ex:irrelevant} are such minimalist, alternative strategy profiles, where double lines still represent choices. Up to symmetry, they constitute from left to right a cycle of improvements that could be intuitively described as lazy, so an actual cycle of length eight can easily be inferred from the short pseudo cycle. This may happen because, although the improvements may look lazy, player $a$ forgets about her choices in a subgame (of the root) when leaving it, and may settle for different choices when coming back to the subgame. This suggests that even counter-factual choices are sometimes relevant. In particular, this means that lazy improvement is not a \emph{natural} dynamics in the sense of \name{Hart} \cite{hart2}; or a \emph{simple} model in the sense of \name{Roth} and \name{Erev} \cite{erev}.

\begin{example}
\label{ex:irrelevant} Let $W$ be winning for player $a$ and $L$ be losing; and vice versa for player $b$. \\
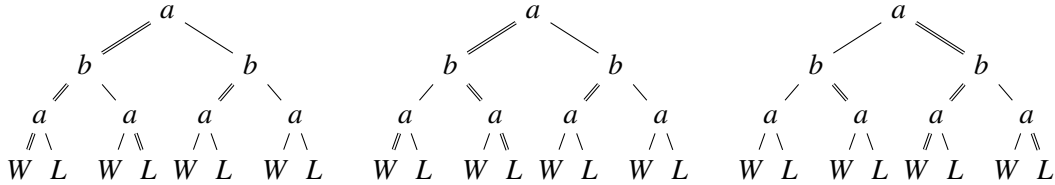
\begin{tabular}{ccc}
\begin{tikzpicture}[level distance=7mm]
\node{a}[sibling distance=22mm]
	child{node{b}[sibling distance=12mm] edge from parent[double]
		child{node{a}[sibling distance=5mm] edge from parent[double]
			child{node{$W$} edge from parent[double]}
			child{node{$L$}}
		}
		child{node{a}[sibling distance=5mm]
			child{node{$W$}}
			child{node{$L$} edge from parent[double]}
		}
	}
	child{node{b}[sibling distance=12mm]
		child{node{a}[sibling distance=5mm] edge from parent[double]
			child{node{$W$}}
			child{node{$L$}}
		}
		child{node{a}[sibling distance=5mm]
			child{node{$W$}}
			child{node{$L$}}
		}
	}
	;
\end{tikzpicture}
&
\begin{tikzpicture}[level distance=7mm]
\node{a}[sibling distance=22mm]
	child{node{b}[sibling distance=12mm] edge from parent[double]
		child{node{a}[sibling distance=5mm]
			child{node{$W$} edge from parent[double]}
			child{node{$L$}}
		}
		child{node{a}[sibling distance=5mm] edge from parent[double]
			child{node{$W$}}
			child{node{$L$} edge from parent[double]}
		}
	}
	child{node{b}[sibling distance=12mm]
		child{node{a}[sibling distance=5mm] edge from parent[double]
			child{node{$W$}}
			child{node{$L$}}
		}
		child{node{a}[sibling distance=5mm]
			child{node{$W$}}
			child{node{$L$}}
		}
	}
	;
\end{tikzpicture}
&
\begin{tikzpicture}[level distance=7mm]
\node{a}[sibling distance=22mm]
	child{node{b}[sibling distance=12mm]
		child{node{a}[sibling distance=5mm]
			child{node{$W$}}
			child{node{$L$}}
		}
		child{node{a}[sibling distance=5mm] edge from parent[double]
			child{node{$W$}}
			child{node{$L$}}
		}
	}
	child{node{b}[sibling distance=12mm] edge from parent[double]
		child{node{a}[sibling distance=5mm] edge from parent[double]
			child{node{$W$} edge from parent[double]}
			child{node{$L$}}
		}
		child{node{a}[sibling distance=5mm]
			child{node{$W$}}
			child{node{$L$} edge from parent[double]}
		}
	}
	;
\end{tikzpicture}
\end{tabular}
\end{example}

The example below shows that for infinite games, a sequence of lazy improvement steps may have multiple accumulation points even for continuous payoff functions; and moreover, that not all accumulation points have to be Nash equilibria.
\begin{example}[{\cite[Example 26]{paulyleroux2}}]
\label{ex:twentysix}
\end{example}

\begin{wrapfigure}{r}{0.5\textwidth}
\begin{tikzpicture}[level distance=8mm]
\node{c}[sibling distance=25mm]
	child{node{$\alpha_0,\beta_0,\gamma_0,\delta_0$} edge from parent[solid,double]}
	child{node{d}[sibling distance=25mm]
		child{node{$\alpha_1,\beta_1,\gamma_1,\delta_1$} edge from parent[double]}
		child{node{c} [sibling distance=25mm] edge from parent[dashed]
		child{node{$\alpha_n,\beta_n,\gamma_n,\delta_n$} edge from parent[double,solid]}
			child{node{} edge from parent[dashed]
				child{node{} edge from parent[draw=none]}
				child{node{$\alpha,\beta,\gamma,\delta$} edge from parent[dashed]}
			}
		}
	};
\end{tikzpicture}
\end{wrapfigure}

Let us consider games with four players $a$, $b$, $c$, and $d$. Given four real-valued sequences $\mathcal{A}=(\alpha_n)_{n\in\mathbb{N}}$, $\mathcal{B}=(\beta_n)_{n\in\mathbb{N}}$, $\mathcal{C}=(\gamma_n)_{n\in\mathbb{N}}$, and $\mathcal{D}=(\delta_n)_{n\in\mathbb{N}}$ converging towards $\alpha$, $\beta$, $\gamma$, and $\delta$, let $T(\mathcal{A},\mathcal{B},\mathcal{C},\mathcal{D})$ be the following game and strategy profile. Note that apart from the payoffs, the underlying game effectively involves players $c$ and $d$ only. If $\mathcal{C}$ and $\mathcal{D}$ are increasing, the lazy improvement dynamics sees players $c$ and $d$ alternating in switching their top left-move to a right-move.

Let $\mathcal{A} := \mathcal{B} := (1+\frac{1}{n+1})_{n\in\mathbb{N}}$ and let $\mathcal{C} := \mathcal{D} := (1-\frac{1}{n+1})_{n\in\mathbb{N}}$. Starting from the profile below, players $c$ and $d$ will continue to unravel the subgame currently chosen jointly by $a$ and $b$. Player $b$ will keep alternating her choices to pick the least-unraveled subgame available to her. Player $a$ will prefer to chose a subgame where player $b$ currently chooses right, and also prefers less-unraveled subgames.

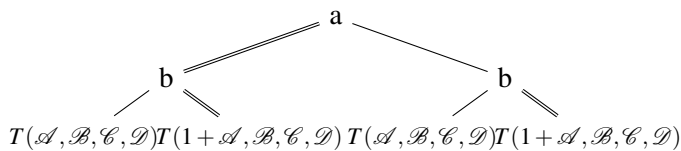
\begin{tikzpicture}[level distance=8mm]
\node{a}[sibling distance=45mm]
	child{node{b}[sibling distance=22mm] edge from parent[double]
		child{node{\footnotesize $T(\mathcal{A},\mathcal{B},\mathcal{C},\mathcal{D})$}}
		child{node{\footnotesize $T(1+\mathcal{A},\mathcal{B},\mathcal{C},\mathcal{D})$} edge from parent[double]}
	}
	child{node{b}[sibling distance=22mm]
		child{node{\footnotesize $T(\mathcal{A},\mathcal{B},\mathcal{C},\mathcal{D})$}}
		child{node{\footnotesize $T(1+\mathcal{A},\mathcal{B},\mathcal{C},\mathcal{D})$} edge from parent[double]}
	};
\end{tikzpicture}

First of all, already the subgame where $b$ moves first demonstrates that the lazy improvement dynamics will not always converge, hence we have to consider accumulation points rather than limit points. For the next feature, note that there is an infinite sequence of lazy
improvement where players $a$ and $b$ (at both nodes that she owns) switch
infinitely often, and where player $a$ switches only when player $b$
chooses the right subgame (on the induced play). Then the following
strategy profile
is an accumulation point, but it is clearly not a Nash equilibrium.

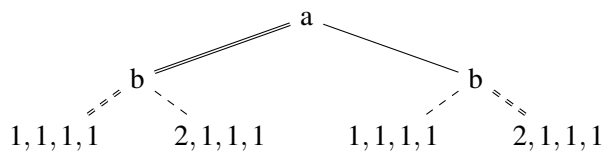
\begin{tikzpicture}[level distance=8mm]
\node{a}[sibling distance=45mm]
	child{node{b}[sibling distance=22mm] edge from parent[double]
		child{node{$1,1,1,1$} edge from parent[double,dashed]}
		child{node{$2,1,1,1$}edge from parent[dashed]}
	}
	child{node{b}[sibling distance=22mm]
		child{node{$1,1,1,1$}edge from parent[dashed]}
		child{node{$2,1,1,1$} edge from parent[double,dashed]}
	};
\end{tikzpicture}

In our current model the players perform lazy improvement updates in a sequential manner. If simultaneity were allowed (yet not compulsory), cycles could occur, as shown in the example below.

\begin{example} It is a cycle up to symmetry only, a proper cycle of length $4$ may be easily derived from it.
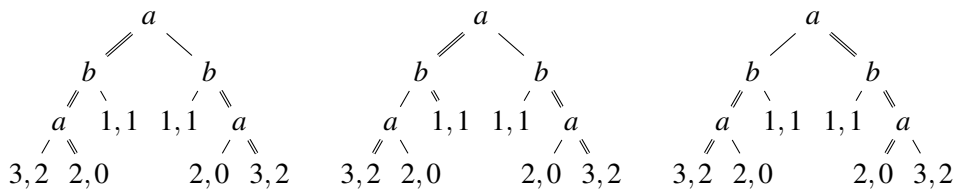
\begin{tabular}{ccccc}
\begin{tikzpicture}[level distance=7mm]
\node{a}[sibling distance=16mm]
	child{node{b}[sibling distance=8mm] edge from parent[double]
		child{node{a}[sibling distance=8mm]edge from parent[double]
			child{node{$3,2$}}
			child{node{$2,0$}edge from parent[double]}
		}
		child{node{$1,1$}[sibling distance=8mm]}
	}
	child{node{b}[sibling distance=8mm]
		child{node{$1,1$}[sibling distance=8mm]}
		child{node{a}[sibling distance=8mm]edge from parent[double]
			child{node{$2,0$}}
			child{node{$3,2$}edge from parent[double]}
		}
	}
	;
\end{tikzpicture}
&
\begin{tikzpicture}[level distance=7mm]
\node{a}[sibling distance=16mm]
	child{node{b}[sibling distance=8mm] edge from parent[double]
		child{node{a}[sibling distance=8mm]
			child{node{$3,2$}edge from parent[double]}
			child{node{$2,0$}}
		}
		child{node{$1,1$}[sibling distance=8mm]edge from parent[double]}
	}
	child{node{b}[sibling distance=8mm]
		child{node{$1,1$}[sibling distance=8mm]}
		child{node{a}[sibling distance=8mm]edge from parent[double]
			child{node{$2,0$}}
			child{node{$3,2$}edge from parent[double]}
		}
	}
	;
\end{tikzpicture}
&
\begin{tikzpicture}[level distance=7mm]
\node{a}[sibling distance=16mm]
	child{node{b}[sibling distance=8mm]
		child{node{a}[sibling distance=8mm]edge from parent[double]
			child{node{$3,2$}edge from parent[double]}
			child{node{$2,0$}}
		}
		child{node{$1,1$}[sibling distance=8mm]}
	}
	child{node{b}[sibling distance=8mm]edge from parent[double]
		child{node{$1,1$}[sibling distance=8mm]}
		child{node{a}[sibling distance=8mm]edge from parent[double]
			child{node{$2,0$}edge from parent[double]}
			child{node{$3,2$}}
		}
	}
	;
\end{tikzpicture}
\end{tabular}
\end{example}

This behaviour can be avoided by considering lazy best-response dynamics, rather than merely lazy better-response. In the sequential case, clearly the termination of the latter implies termination of the former. In the simultaneous case, we have:

\begin{proposition}\label{prop:synclazybestresp}
The synchronous lazy best-response sequences in a game with $n$ internal nodes have length at most $2^n$, provided that the players have acyclic preferences.

\begin{proof}
It suffices to prove the claim for preferences that are linear orders, which we prove by induction on the number of internal nodes of the game $g$. (It holds for zero.) Let $v$ be an internal node in $g$ whose children are all leaves, let $a$ be the owner of $v$, and let us consider a sequence where $a$ always chooses the same outcome $x$ at $v$. Let $g'$ be the game derived from $g$ by replacing $v$ with a leaf enclosing the outcome $x$. The synchronous lazy best-response sequence in $g$ corresponds, by restriction of the profiles, to a sequence in $g'$, so it has length at most $2^{n-1}$ by I.H. Now let us consider an arbitrary sequence, and note that $a$ can change choices only once at $v$, from some non-preferred outcome to her preferred one (among the outcomes occurring below $v$). So the length of a sequence in $g$ is at most $2^{n-1} + 2^{n-1} = 2^{n}$.
\end{proof}
\end{proposition}

\section*{Acknowledgements}
This work benefited from the Royal Society International Exchange Grant IE111233 (while Le Roux was at the TU Darmstadt and Pauly at the University of Cambridge). The authors were partially supported by the ERC inVEST (279499)
project.

We are grateful to Dietmar Berwanger, Victor Poupet, and Martin Ziegler for helpful discussions, and to an anonymous referee for his or her helpful comments on a previous version of this paper.

\bibliographystyle{eptcs}
\bibliography{../../../spieltheorie}

\end{document}